\documentclass[twocolumn,amsmath,amssymb, pra]{revtex4-1}

\usepackage{graphicx}
\usepackage{float}
\usepackage{amsfonts,amstext}
\usepackage{url}
\usepackage{color}
\usepackage{braket}
\usepackage{titlesec}
\usepackage{soul}
\usepackage{amsthm}
\usepackage{physics}
\usepackage{hhline}
\usepackage{bbm}

\makeatletter
\def\@seccntformat#1{\@ifundefined{#1@cntformat}%
	{\csname the#1\endcsname\quad}  
	{\csname #1@cntformat\endcsname}
}
\let\oldappendix\appendix 
\renewcommand\appendix{%
	\oldappendix
	\newcommand{\section@cntformat}{\appendixname~\thesection: \ }
}
\makeatother

\titleformat{\subsection}[runin]{\it}{}{}{}
\titlespacing{\subsection}{10pt}{10pt}{5pt}

\def\q0{\underline{0}}

\def\H{{\cal H}}

\def\S{\mathcal S}

\def\id{{\mathbbm 1}}

\def\H{{\cal H}}
\def\U{{\mathcal U}}
\def\R{\mathcal{R}}

\def\tr{\text{tr}}
\def\proj#1{\ket{#1}\!\bra{#1}}

\newtheorem{theo}{Theorem}

\newtheorem{prop}[theo]{Proposition}

\newtheorem{sdp}[theo]{SDP}

\renewcommand{\eqref}[1]{Eq.(\ref{#1})}
\newcommand{\figref}[1]{Fig.\ref{#1}}
\newcommand{\tableref}[1]{Tab.(\ref{#1})}
\newcommand{\sdpref}[1]{SDP \ref{#1}}

\newcommand{\be}{\begin{equation}}
\newcommand{\ee}{\end{equation}}

\begin{document}

\title{Entanglement Detection Beyond Measuring Fidelities}
\author{M. Weilenmann, B. Dive,  D. Trillo, E. A.  Aguilar,  M. Navascu\'es}
\affiliation{Institute for Quantum Optics and Quantum Information (IQOQI), 3 Boltzmanngasse Vienna 1090, Austrian Academy of Sciences}

\date{\today}

\begin{abstract}
One of the most widespread methods to determine if a quantum state is entangled, or to quantify its entanglement dimensionality, is by measuring its fidelity with respect to a pure state.  In this Letter we find a large class of states whose entanglement cannot be detected in this manner; we call them \emph{unfaithful}. We find that unfaithful states are ubiquitous in information theory.
For small dimensions, we check numerically that most bipartite states are both entangled and unfaithful. Similarly, numerical searches in low dimensions show that most pure entangled states remain entangled but become unfaithful when a certain amount of white noise is added. We also find that faithfulness can be self-activated, i.e., there exist instances of unfaithful states whose tensor powers are faithful. To explore how the fidelity approach limits the quantification of entanglement dimensionality, we generalize the notion of an unfaithful state to that of a \emph{$D$-unfaithful state}, one that cannot be certified as $D$-dimensionally entangled by measuring its fidelity with respect to a pure state. For describing such states, we additionally introduce a hierarchy of semidefinite programming relaxations that fully characterizes the set of states of Schmidt rank at most $D$.
\end{abstract}

\maketitle

Entanglement is a fundamental aspect of quantum information and one of the key dividing factors between the quantum and the classical worlds. This is shown by the wide range of protocols, such as teleportation \cite{teleport}, device independent quantum key distribution \cite{diqkd0, diqkd1, diqkd2}, one-way quantum computation \cite{measurement_based_comp}, and metrology \cite{metrology2, metrology1}, in which it is a necessary resource. The amount of entanglement, which can be quantified, for example, by the entanglement dimension, is an important factor for applications~\cite{Vidal2003, Jozsa2003, Wang2005, Barreiro2005, Lloyd2008, Lanyon2009, Zhang2013, Huber2013, Nest2013}. A larger entanglement dimension allows, for instance, for the encoding of information in a larger number of entangled degrees of freedom and can improve the tolerance to noise in cryptographic protocols~\cite{Cerf, Bruss}.

As entanglement is such a useful resource, it is important to have experimental and theoretical tools to detect it \cite{entanglementdetection_toth, Friis2019}, and for these to correctly identify entanglement even in noisy states. Perhaps the most commonly used method to detect entanglement is via linear witnesses, but their characterization has been proven a hard problem \cite{NP_hard1, NP_hard2}, even though there are some general methods to tackle it \cite{DPS1, DPS2, DPS3}. The situation is substantially harder still when we want to detect the entanglement dimension of a state, as a practical characterization of the $D$-positive maps necessary for this task is still missing (note that one of the methods for entanglement quantification proposed in \cite{Toth2015} can be used to characterize $D$-positive maps, but at a prohibitive computational cost). In part due to these difficulties, it is common practice to detect the entanglement or entanglement dimension of a state via its fidelity with respect to a pure reference state~\cite{Flammia2011fid, Steinlechner2017fid, Bavaresco2018fid, Erhard2018fid}. The question we address in this Letter is whether such a method misses out many instances of entangled states and, if so, what properties these states have. In addition, we address these questions also for the detection of multi-dimensional entanglement.

We find, surprisingly, that almost no bipartite entangled states can be detected via fidelities with pure states. We call such states \emph{unfaithful}. Using a simple semidefinite programming (SDP) ansatz~\cite{sdp} of unfaithful states, we prove that even states as innocent as a mixture between a pure entangled state and the maximally mixed state can be unfaithful. We also show that faithfulness can be self-activated, namely, there exist unfaithful states whose bipartite entanglement can be detected via fidelities when taking their tensor power. Going beyond separability, we extend the concept of unfaithful states to those states whose entanglement dimension cannot be certified with fidelity witnesses. Lacking efficient general tools to determine their entanglement dimension, we introduce a complete hierarchy of semidefinite programming relaxations of the set of all states with Schmidt rank at most $D$.

\subsection*{Definitions.}
A bipartite state $\rho_{AB}$ is \emph{separable} if there exists a probability distribution $\{p_i\}_i$ and states $\{\ket{\psi_i},\ket{\phi_i} \}_i$ such that

\be
\rho_{AB}=\sum_i p_i \proj{\psi_i}_A\otimes\proj{\phi_i}_B.
\ee
\noindent Otherwise, $\rho_{AB}$ is said to be \emph{entangled}. We denote the set of separable quantum states  by $\S$.

The usefulness of a state for many information processing tasks is closely related to its entanglement dimensionality. A mixed bipartite state $\rho_{AB}$ is said to have Schmidt rank at most $D$~\cite{Terhal1999} if there exists a decomposition
\begin{equation}
\rho_{AB}= \sum p_i \proj{\psi_i}_{AB},
\end{equation} 
with $p_i \geq 0$, $\sum_i p_i =1$ and all $\ket{\psi_i}$ have Schmidt rank at most $D$, i.e., we can write each $\ket{\psi_i}$ as $\ket{\psi_i}=\sum_{j=1}^D \sqrt{\lambda_j(\ket{\psi_i})} \ket{\phi_j}_A\ket{\xi_j}_B$, where $(\sqrt{\lambda_j(\ket{\psi_i})} )_j$ are called the Schmidt coefficients of $\ket{\psi_i}$ and $\{ \ket{\phi_j} \}_j$ and $\{ \ket{\xi_j} \}_j$ are sets of orthonormal vectors. We denote by $\S_D$ the set of all such states $\rho_{AB}$. Note that all these sets are convex, $\S_1=\S$, and $\S_D\subseteq \S_{D+1}$. If $\rho_{AB} \not\in \S_{D-1}$, we say that $\rho_{AB}$ is \emph{$D$-dimensionally entangled}. If, additionally, $\rho_{AB} \in \S_{D}$ we say that it has \emph{Schmidt rank $D$}. 
Tying this back to the definition of simple entanglement, all entangled states are $2$-entangled, and have Schmidt rank at least $2$.

An \emph{entanglement witness} $W$ is a Hermitian operator with the property that $\tr[W\sigma_{AB}] \geq 0$ for all separable states $\sigma_{AB} \in \S$. Hence, a measurement on a state $\rho_{AB}$ having a negative expectation value $\tr[W\rho_{AB}]<0$, implies that $\rho_{AB}$ is entangled. Since $\S$ is convex, the hyperplane separation theorem implies that for any entangled state $\rho_{AB}$ there exists a witness $W$ such that $\tr[W \rho_{AB}]<0$. The same considerations apply when trying to witness entanglement dimension: for any state $\rho_{AB} \not\in \S_{D-1}$ there exists a \emph{dimension-$D$-witness} $W_{D}$ such that for all $\sigma_{AB} \in \S_{D-1}$, $\tr[W_{D} \sigma_{AB}] \geq 0$ and $\tr[W_{D}\rho_{AB}]<0$. In this terminology, an entanglement witness is a dimension-$2$-witness.

One commonly used form of a witness, both for entanglement and $D$-dimensional entanglement, is the so-called pure fidelity witness:
\begin{equation}
W_{D} := F \id - \ket{\psi_W}\bra{\psi_W},
\label{eq:fidelitywitness}
\end{equation}
where $\ket{\psi_W}$ is a fixed entangled state and $F$ a real number. In order to witness as many entangled states as possible it is desirable to have $F$ be as small as possible. The minimum value it can have while still satisfying $\tr[W_D \sigma_{AB}] \geq 0$ for all $\sigma_{AB} \in \S_{D-1}$ is $\min F = \sum_{i=1}^{D-1}\lambda_i(\ket{\psi_W})$ \cite{Bourennane2004}, where $(\sqrt{\lambda_i(\ket{\psi_W})} )_i$ are the Schmidt coefficients in non-increasing order. 
Such a witness detects a state $\rho_{AB}$ as being $D$-dimensionally entangled if $\bra{\psi_W} \rho_{AB} \ket{\psi_W} > F$. In other words, $\rho_{AB}$ is certified as being $D$-dimensionally entangled if the overlap with the target state $\ket{\psi_W}$ is big enough. As intuition suggests, this works well when $\rho_{AB}$ is close to the (usually highly entangled) state $\ket{\psi_W}\bra{\psi_W}$. However, due to the frequency with which witnesses of the type of \eqref{eq:fidelitywitness} are used~\cite{Flammia2011fid, Steinlechner2017fid, Bavaresco2018fid, Erhard2018fid}, the question naturally arises whether for some $D$ \emph{all} $D$-dimensionally entangled states $\rho_{AB}$ can be certified with such a class of witnesses and, if not, whether such witnesses are generally useful. The broadest interest is in answering the latter question for usual entanglement, meaning for $D=2$. 

We say that $\rho_{AB}$ is \emph{$D$-unfaithful} if it satisfies $\tr[W_D\rho_{AB}]\geq 0$ for all witnesses $W_D$ of the form of \eqref{eq:fidelitywitness}, and denote the set of such states by $\U_D$. In the same way as we call 2-entangled states simply entangled states, $2$-unfaithful is referred to as \emph{unfaithful}. The interesting set of states are those in $\U_D \setminus \S_{D-1}$, which have $D$-dimensional entanglement that cannot be seen with fidelities to pure states. The relations between these sets are shown in \figref{fig:egg}, in particular, we can see the inclusion relationships $\S_{D-2} \subseteq \S_{D-1}\subseteq \U_D\subseteq \U_{D+1}$. Note, however, that there is no clear relationship between $\S_D$ and $\U_{D'}$ for $D \ge D'$.

\begin{figure}[h]
\centering
\includegraphics[width=1\columnwidth]{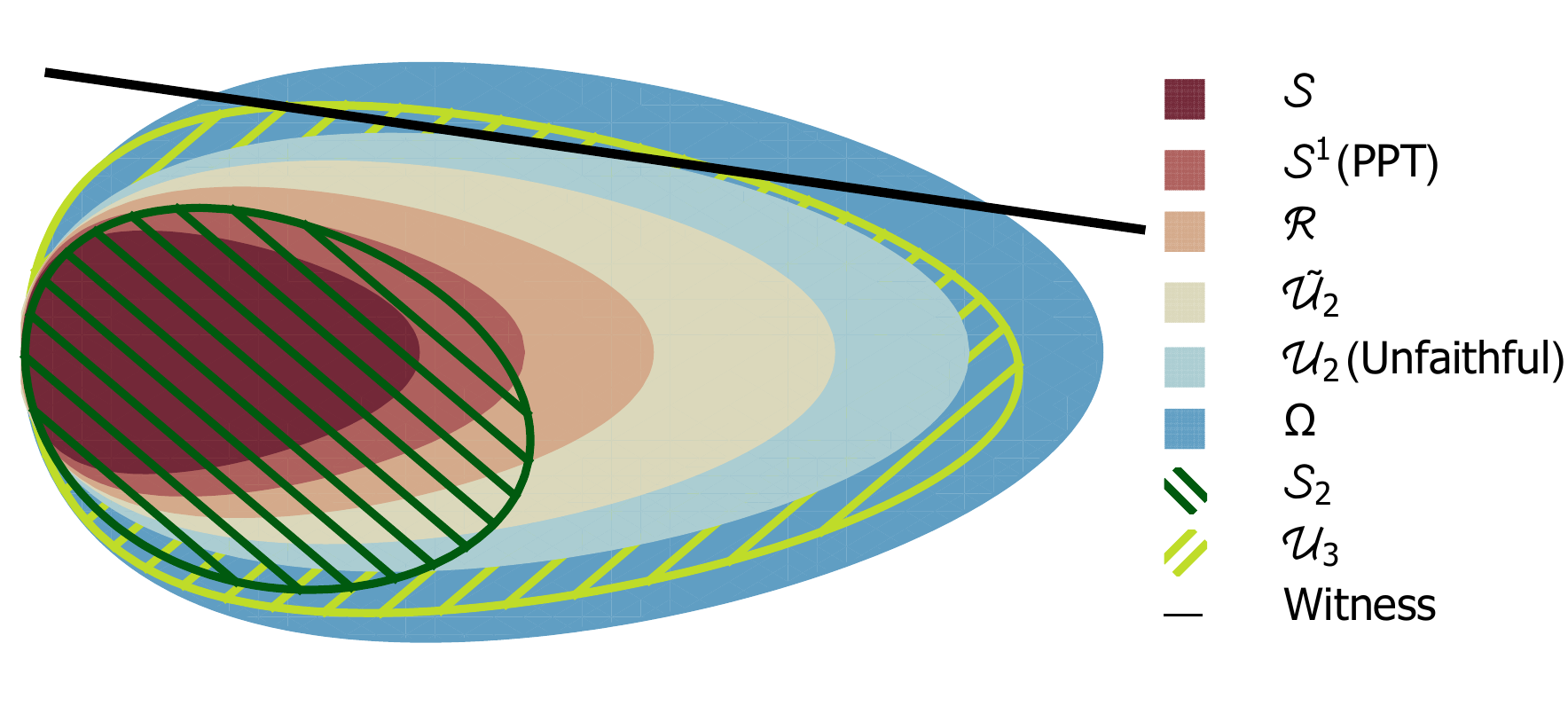}
\caption{A schematic picture of the relationship between the sets of states investigated. The large gap between the pure-state fidelity witness and $\S$ shows how poorly such witnesses can perform. The three intermediate sets $\S^1$, $\R$, and $\widetilde{\U}_2$ are easier to compute approximations of $\S$ (for the first two) and of $\U_2$ (for the latter). All the same objects can be drawn for $D$-dimensional entanglement, for clarity, only $\S_2$ and $\U_3$ are shown (with stripes). $\Omega$ labels the set of all bipartite quantum states.}
\label{fig:egg}
\end{figure}

\subsection*{Approximations with semidefinite programming.} Our aim is to analyse which $D$-dimensionally entangled states, for some $D\geq 2$, are uncertifiable by dimension-$D$-witnesses, that is those in $( \U_D \setminus \S_{D-1})_{D\geq 2}$. We do this by constructing semidefinite programs that form an outer approximation of $\S_{D-1}$ (\sdpref{sdp:one}) and an inner approximation of $\U_D$ (\sdpref{sdp:two}).

Let us begin with the former. For the set $\S_1$, we are concerned with  detecting whether a state is entangled. There exist several complete criteria for entanglement detection (see, e.g., \cite{Eisert}). We use the Doherty-Parrilo-Spedalieri (DPS) hierarchy \cite{DPS1,DPS2,DPS3}, which is an infinite nested sequence of sets $\S^1\supseteq \S^2\supseteq \S^3\supseteq \cdots$, with the property that $\lim_{k\to\infty} \S^k = \S$, see Appendix~\ref{app1} for a short description. Furthermore, each of these sets is characterizable via semidefinite programming. This means that, if $\rho_{AB}$ is entangled at all, then there exists an SDP that will detect this. The first set $\S^1$ of the DPS hierarchy is the set of states with positive partial transpose~\cite{PPT1,PPT2}, which we call PPT in the following.

Here, we introduce a family of hierarchies of SDPs that generalizes the DPS criterion to general $D$ by converging to $\S_{D}$ from the outside:

\begin{sdp}
	\label{sdp:one}
	Let $\sigma_{AB}$ be a bipartite state, with local dimensions $d_A$ and $d_B$ respectively. We say that $\sigma_{AB} \in \S_D^k$, if there exists a positive semidefinite matrix $\omega$ on the subsystems $AA'B'B$ such that
\begin{align}
	&\frac{\omega}{D} \in \S^k, \quad \Pi^\dagger_{D}\omega \Pi_{D}=\sigma_{AB} \nonumber\\
	&\Pi_D = \id_A\otimes\ket{\psi^+_{D}}_{A'B'}\otimes \id_B 
	\label{newHier}
\end{align}
where the dimensions of $A'$ and $B'$ are both $D$, $\ket{\psi^+_{D}}_{A'B'}:= \sum_{j=1}^{D}\ket{j,j}$ is the non-normalized maximally entangled state in this dimension and the bipartition of $\omega$ relevant for the definition of $\S^k$ is $AA'|B'B$. 
\end{sdp}
\noindent Note that the first condition in (\ref{newHier}) implies $\tr[\omega]=D$.
 For $D=1$, \sdpref{sdp:one} coincides with the DPS hierarchy for characterising separable states. From the above definition, it is straightforward that $\S_{D}^k\supseteq \S_{D}^{k+1}$ and that each of these sets can be characterized by an SDP.  Furthermore, the hierarchy converges, i.e., $\S_{D}\subseteq \S^k_{D}$ for all $k, D$ and $\lim_{k\to\infty} \S^k_{D}=\S_{D}$, which we prove in Appendix~\ref{app2}. More precisely, we prove that for any $\sigma_{AB}\in \S_D^k$, there exists $\tilde{\sigma}_{AB}\in \S_D$ such that $\|\sigma_{AB}-\tilde{\sigma}_{AB}\|_1\leq O\left(\frac{d^2D^4}{k^2}\right)$, where $\| \cdot \|_1$ is the trace norm and $d=\min \{d_A,d_B \}$. For an implementation of the lowest order of the hierarchy (\ref{newHier}) in Python, we refer the reader to~\cite{GitCode}.

A different approach to constructing an SDP that witnesses the entanglement dimension was introduced in \cite{Toth2015}. While both our method and theirs use the DPS hierarchy as a way approximating the separable set that is optimized over, their method additionally requires as many copies of the state as the Schmidt rank being witnessed. This makes our approach greatly more memory efficient.

For the problem of certifying that a state is in $\U_D$ we utilize an inner approximation that can be realized with an SDP.  We do this by defining a new set, $\widetilde{\U}_D$, according to the following:

\begin{sdp}
\label{sdp:two}
	Let $\rho_{AB}$ be a bipartite state. If there exists $\mu\in [0,1]$ and positive semidefinite operators $M_A$, $M_B$ such that 
	\begin{align} 
	  &M_A \otimes \id_B + \id_A \otimes M_B \geq \rho_{AB} \nonumber \\
	&\mu (D-1) = \tr[M_A], \hspace{12mm} \mu\id-M_A\geq 0,\nonumber\\
	&(1-\mu)(D-1)=\tr[M_B], \hspace{2.5mm}(1-\mu)\id-M_B\geq 0,
	\label{criter}
	\end{align}
\noindent then we say that $\rho_{AB}\in \widetilde{\U}_D$.
\end{sdp}

\noindent This set is an inner approximation to $\U_D$. That is, $\widetilde{\U}_D \subseteq \U_D$, which we prove in Appendix~\ref{app3}.

When $D=2$, \sdpref{sdp:two} is a generalization of the \emph{reduction criterion}, which states that either $\id_A\otimes \rho_B-\rho_{AB} \geq 0$, or $\rho_A \otimes \id_B-\rho_{AB} \geq 0$ holds. For two qubits or a qubit and a qutrit this is also equivalent to being PPT \cite{reduction1, reduction2}. We denote by $\mathcal{R}$ the set of all states that satisfy the reduction criterion. It follows that $\mathcal{R}\subseteq \widetilde{\U}_2 \subseteq \U_2$. That any reducible state is unfaithful was already proven in \cite{Piani}; note however that our findings in \tableref{fig:sampling} imply that $\widetilde{\U}_2$ is strictly larger than the set of reducible states, and therefore a better approximation of $\U_2$. An illustration of $\S_D^k, \R, \text{ and } \widetilde{\U}_D$, along with their relation to the previously defined sets, is shown in \figref{fig:egg}.

\subsection*{Most states are unfaithful.} Armed with the above two technical tools, we next argue that entangled unfaithful states are, in fact, ubiquitous in quantum information theory. To this end, we sample random bipartite states according to the Hilbert-Schmidt~\cite{Zyczkowski1,Zyczkowski2} and the Bures~\cite{Osipov} measures. For completeness we briefly describe these sampling techniques in Appendix~\ref{app4}. For bipartite systems of local dimension $d$ on both subsystems, we find (with both measures) that the larger $d$, the larger the fraction of unfaithful states (see \tableref{fig:sampling}). In fact, from $d=3$ on, we find that most states are entangled, as detected by them not having a positive partial transpose ($\rho_{AB} \notin \S^1$), and are certified as unfaithful by \sdpref{sdp:two}. From $d=5$ on, we find that essentially all states we generated are entangled but at the same time unfaithful, regardless of whether we sample according to the Hilbert-Schmidt or the Bures metric. This highlights the importance of characterising the entanglement of unfaithful qudit states beyond relying on fidelity witnesses.

\begin{table}[h]
	\begin{tabular}{|c|c|c|c|c|}
		\hline
		$d$ 	&\hspace{0.5mm} $\S^1$ (HS) \hspace{0.5mm}  &\hspace{0.5mm} $\widetilde{\U}_2 \setminus \S^1$ (HS) \hspace{0.5mm}  & \hspace{0.5mm} $\S^1$ (B) \hspace{0.5mm}  &\hspace{0.5mm} $\widetilde{\U}_2 \setminus \S^1$ (B) \\
		\hline 
		$2$ & $24.2 \%$ & $21.2 \%$  & $7.4 \%$ & $15.4 \%$   \\
		$3$ &	$0.01 \%$ & $94.5 \%$ & $0 \%$ & $54.8 \%$ \\
		$4$ &	$0 \%$ & $100 \%$ & $0 \%$ & $97.0 \%$ \\
		$5$ &	$0 \%$ & $100 \%$  & $0 \%$ & $100 \%$ \\
		\hline
	\end{tabular}
\caption{Percentage of unfaithful states detected via \sdpref{sdp:two} and the PPT condition when sampling uniformly according to either the Hilbert-Schmidt metric (HS) or the Bures metric (B) (sampling $10^6$ points for each dimension in each case). $d$ denotes the local dimension of each subsystem of the bipartite system. Note that we are actually interested in $\U_2 \setminus \S \supseteq \widetilde{\U}_2 \setminus \S^1$, which may be even larger.
}
\label{fig:sampling}
\end{table}

For $D\geq 2$, random sampling suggests that the set $\S_D^1$ almost exhausts the set of quantum states up to at least local dimension $d=4$. This assertion is based on sampling $1000$ random states in $d=3$ (sampled according to Hilbert-Schmidt as well as Bures metric), which were all found to belong to $\S_2^1$. In $d=4$, we found the same when sampling $1000$ real states, see~\cite{GitCode} for details regarding the implementation and recall that the respective techniques are described in Appendix~\ref{app4}. 
Nevertheless, there are states in $\widetilde{U}_2 \setminus \S_2^1$, i.e., states that are at least $3$-dimensionally entangled but unfaithful. Constructions of such states can for instance be found in~\cite{Chen2017, Huber2018}.

For $D\geq 2$, there are furthermore examples of states that are faithful but $3$-unfaithful. Consider, for instance, 
\begin{equation}\label{eqn:Mirjams state}
\frac{1}{2}\ketbra{\Psi_3}{\Psi_3} +\frac{1}{2}\left (\frac{\ket{23}+\ket{32}}{\sqrt{2}}\right)\left(\frac{\bra{23}+\bra{32}}{\sqrt{2}}\right),
\end{equation}
where $\ket{\Psi_3}:= \tfrac{1}{\sqrt{3}} (\ket{00}+\ket{11}+\ket{22})$. This state belongs to $\widetilde{\U}_3 \setminus (\S_2^1  \cup  \U_2)$. That is, it is a $3$-unfaithful state that is indeed of Schmidt rank $3$, but that is $2$-faithful. Note that for the last claim we cannot use our SDP approximation because it is inner. However, it is easy to find a fidelity witness that certifies its entanglement; namely
\begin{equation*}
	W_2:=\frac{1}{3}\id - \ketbra{\Psi_3}{\Psi_3}.
\end{equation*}
For the certification of the Schmidt rank as well as the $2$-unfaithfulness of the state using \sdpref{sdp:one} (where the first level of the hierarchy was sufficient) and~\sdpref{sdp:two}, we refer the reader to~\cite{GitCode}.

\subsection*{Noisy pure states are unfaithful.} Having established that most states are unfaithful, we now investigate if this holds for states that are commonly used and important in quantum information. Consider the Bell state $\ket{\Psi_2} = 1 / \sqrt{2} (\ket{00}+\ket{11})$ embedded in a $d\times d$ dimensional bipartite Hilbert space and define the mixed state 
\begin{equation}\label{eqn: noisy state}
\rho_{AB}(p) =p\frac{ \id_{AB}}{d^2} + (1-p)\ketbra{\Psi_2}{\Psi_2} .
\end{equation} 
For $d>2$, there exists a large parameter regime, illustrated in \figref{fig:PureWithNoise}, where this state is entangled but unfaithful. This is not a specific property of $\ket{\Psi_2}$: we found numerically (with $10^6$ random Haar-distributed pure states for each of $d=3, 4, 5$) that all pure entangled states are unfaithful and entangled, for a certain range of white noise. The only exception we know of is the maximally entangled state $1/\sqrt{d} \sum_{i=0}^{d-1}\ket{ii}$ subjected to white noise in a Hilbert space of dimension $d\times d$.

\begin{figure}[h]
\centering
\includegraphics[width=1\columnwidth]{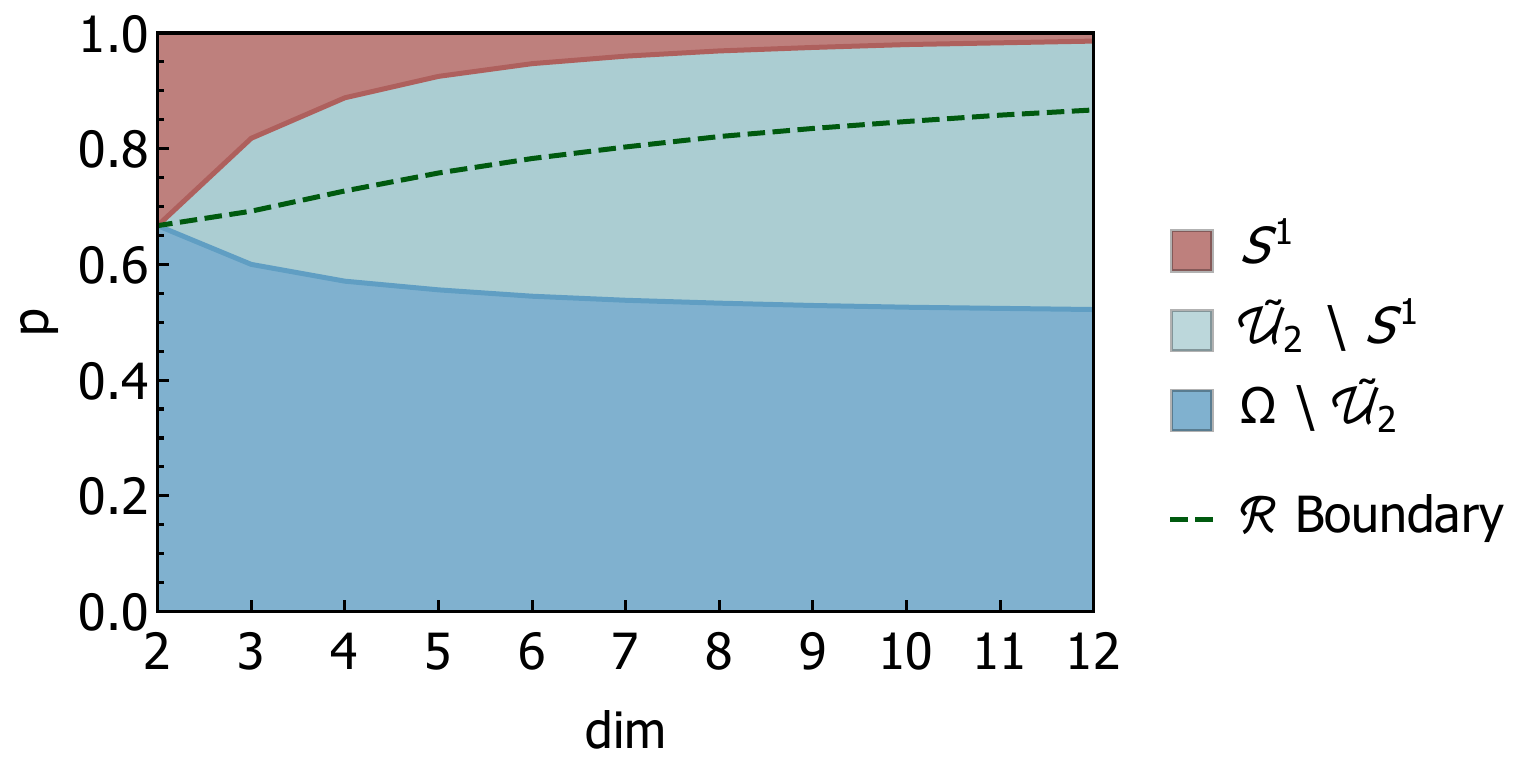}
\caption{Graph of the unfaithful-but-entangled region of the family of states $\rho_{AB}(p)$ from \eqref{eqn: noisy state}. Along the x-axis we display the local dimension $d$ of the Hilbert space, along the y-axis the proportion of white noise $p$. The bottom region of the graph indicates that, as $d$ increases, the ability of fidelity witnesses to detect entanglement becomes less tolerant to noise: the maximal $p$ such that $\rho_{AB} (p) \in \Omega \setminus \widetilde{\U}_2$ decreases with $d$, where $\Omega$ denotes the set of all bipartite states and $\widetilde{\U}_2$ is the approximation to $\U_2$ defined by \sdpref{sdp:two}. At the same time, the top region shows that, in high dimensions, states are entangled (as certified by them having a negative partial transpose, $\rho_{AB}(p) \notin S^1$) even at extremely high noise levels. Thus in these dimensions states are unfaithful and entangled in a large parameter regime. To illustrate its difference from \sdpref{sdp:two}, the reducibility criterion is plotted as a dotted line; only the states above it are reducible. Note that the reducibility and the PPT boundaries can be computed analytically. }
\label{fig:PureWithNoise}
\end{figure}

For $3$-unfaithfulness, we observe an analogous behaviour: the embedded state $\ket{\Psi_3}:= 1/\sqrt{3} (\ket{00}+\ket{11}+\ket{22})$ mixed with noise as in \eqref{eqn: noisy state} belongs to $\widetilde{\U}_3 \setminus \S_2^1$ (meaning that it is certified to have Schmidt rank $3$ but that this cannot be detected with a dimension-$3$-witness) for $p$ in the following ranges: in $d=4$ for $p\in (0.364, 0.449)$; in $d=5$ for $p\in (0.357,0.493)$. In $d=3$, the state becomes $3$-unfaithful at the same point as we cease to certify (using $\S_2^1$) that it is $3$-dimensionally entangled.

\subsection*{Faithfulness can be self-activated.} 
Note that, if $\rho_{AB}$ is faithful, then so is $\rho^{\otimes n}_{AB}$. Indeed, recall (from \eqref{eq:fidelitywitness} and~\cite{Bourennane2004}) that for a faithful $\rho_{AB}$ there exists a $\ket{\psi}$, such that  $\bra{\psi}\rho_{AB}\ket{\psi}>\lambda_1(\ket{\psi})$, where $\sqrt{\lambda_1(\ket{\psi})}$ is the maximal Schmidt coefficient of $\ket{\psi}$. Then, 
\begin{equation}
\bra{\psi}^{\otimes n}\rho^{\otimes n}_{AB}\ket{\psi}^{\otimes n}>\lambda_1(\ket{\psi})^n=\lambda_1(\ket{\psi}^{\otimes n}),
\end{equation}
which implies that $\lambda_1(\ket{\psi}^{\otimes n})- \ket{\psi}^{\otimes n} \bra{\psi}^{\otimes n}$ witnesses the entanglement of $\rho^{\otimes n}_{AB}$.
The property of being faithful is thus preserved under tensor powers. It is therefore natural to ask if unfaithfulness is also preserved this way, or if faithfulness can be \emph{self-activated}.

Such a self-activation effect is impossible for states in $\R$. Indeed, whenever two matrices satisfy $A \geq B \geq 0$ we have $A\otimes A \geq B\otimes B \geq 0$, so that
\begin{equation}
\rho_A \otimes \id_B -\rho_{AB} \geq 0\Rightarrow (\rho_A \otimes \id_B)^{\otimes n} -\rho^{\otimes n}_{AB} \geq 0  .
\end{equation}

Nevertheless, it is possible to self-activate the faithfulness of certain states. Let us introduce the states
\begin{align}
\ket{\phi_1} := &\; 0.628\ket{11} - 0.778 \ket{22} \\
\ket{\phi_2} := &\; 0.807\ket{01}-0.185 \ket{02} - 0.102 \ket{10} -0.027 \ket{11}  \nonumber \\
  + &\; 0.011 \ket{12} + 0.551 \ket{20} -0.024 \ket{21} - 0.022 \ket{22}. \nonumber
\end{align}
Then, the state
\begin{equation*}
\rho_{AB} := 0.999( 0.50179 \ketbra{\phi_1}{\phi_1} + 0.49821 \ketbra{\phi_2}{\phi_2})+ 0.001 \frac{\id}{9}
\end{equation*}
is unfaithful, but $\rho_{AB}^{\otimes 2}$ is faithful over the partition $A|B$. This is proved, and the witness showing faithfulness given, in \cite{GitCode}. Thus, fidelity witnesses are sometimes useful, even if a target state is unfaithful.

\subsection*{Conclusion.} In this Letter, we have introduced the set of unfaithful states, namely, those states whose entanglement cannot be detected via fidelities to pure states. We analyzed different properties of such states: their frequency, their robustness and the phenomenon of activation. As we discovered, the set of unfaithful states is large and comprises quantum states which are very relevant in quantum information theory. We argued, based on our studies for small $D$, that the sets $\U_D$ with $D>2$ behave similarly to the  unfaithful states. Our work is therefore to be understood as a warning towards the blind application of fidelity-type entanglement witnesses. 

While there exist several general methods to detect entanglement beyond state fidelities, one cannot say the same regarding the quantification of the entanglement dimensionality. In this regard, our work also provides methods to construct general dimension-$D$-witnesses, by means of a complete hierarchy of semidefinite programs. As this hierarchy is substantially more memory efficient than prior methods, we expect it to play a significant role in the quantification of entanglement in noisy experimental setups.

There are, naturally, some open questions about unfaithful states that could be studied in future work. One avenue is to extend the results from bipartite entanglement to multi-partite entanglement, which has a much richer and more elaborate structure. A different path is to investigate activation further. We have shown that some (but not all) unfaithful states become faithful when we take a tensor power of them; this leads to the hypothesis that for any entangled and unfaithful state, there exist another unfaithful state such that their tensor product is faithful. If this is true, then pure-state fidelities could always detect entanglement provided that we have access to the correct auxiliary state.

\subsection*{Acknowledgements.} We thank Otfried Gühne for pointing out errors in Table I in the previous version of our article. This work was supported by the
Austrian Science fund (FWF) stand-alone project P~30947.

\onecolumngrid

\appendix

\section{The Doherty-Parrilo-Spedalieri hierarchy} \label{app1}
The Doherty-Parrilo-Spedalieri (DPS) hierarchy \cite{DPS1,DPS2,DPS3} stems from the following observation: let $\rho_{AB}$ be a separable state, with  decomposition 

\be
\rho_{AB}=\sum_i p_i\proj{\phi_i}\otimes\proj{\psi_i}.
\ee

\noindent Then, the states $\rho_{AB_1...B_k}$ defined by

\be
\rho_{AB_1...B_k} =\sum_i p_i\proj{\phi_i}\otimes\proj{\psi_i}^{\otimes k}
\ee
\noindent satisfy:

\begin{enumerate}
	\item
	$\rho_{AB_1...B_k}$ is positive semidefinite and PPT with respect to any bipartition of the systems $A, B_1,B_2,...B_k$.
	\item
	$\rho_{AB_1...B_k}$ has support on the space $\H_A\otimes\H_{sym}^{k,d_B}$, where $\H_A$ is $A$'s Hilbert space and $\H_{sym}^{k,d_B}$ is the symmetric space of $k$ $d_B$-dimensional systems, with $d_B=\mbox{dim}(\H_B)$.
	\item
	$\tr_{B_1,...,B_{k-1}}[\rho_{AB_1...B_k}] = \rho_{AB}$.
\end{enumerate}

For any state $\rho_{AB}$, any state $\rho_{AB_1...B_k}$ satisfying the three conditions above is called a PPT $k$-symmetric extension of $\rho_{AB}$ on system $B$. As we have seen, if $\rho_{AB}$ is separable, any such extension must necessarily exist. Conversely, as proven in \cite{DPS2}, if any such extension exists for all natural numbers $k$, then $\rho_{AB}$ can be proven to be separable. Note that the same considerations apply to extensions on system $A$, i.e., in principle we could have considered extensions of the type $\rho_{A_1,...,A_kB}$. In practice, to save computational resources, one considers extensions of the system with the smallest dimensionality. 

Now, define $\S^k$ to be the set of all bipartite states admitting a PPT $k$-symmetric extension on the subsystem with the smallest dimensionality. From all the above, it follows that $\S^1\supseteq \S^2\supseteq \ldots \supseteq \S$ and $\lim_{k\to\infty}\S^k = \S$. Verifying that $\rho_{AB}\in S^k$ can be cast as a semidefinite program~\cite{sdp}.

\section{Proof of convergence of the semidefinite hierarchy SDP~1} \label{app2}

\begin{sdp}
	\label{sdp:one}
	Let $\sigma_{AB}$ be a bipartite state, with local dimensions $d_A$ and $d_B$ respectively. We say that $\sigma_{AB} \in \S_D^k$, if there exists a positive semidefinite matrix $\omega$ on the subsystems $AA'B'B$ such that
	\begin{align}
		&\frac{\omega}{D} \in \S^k, \quad \Pi^\dagger_{D}\omega \Pi_{D}=\sigma_{AB} \nonumber\\
		&\Pi_D = \id_A\otimes\ket{\psi^+_{D}}_{A'B'}\otimes \id_B 
		\label{newHier}
	\end{align}
	where the dimensions of $A'$ and $B'$ are both $D$, $\ket{\psi^+_{D}}_{A'B'}:= \sum_{j=1}^{D}\ket{j,j}$ is the non-normalized maximally entangled state in this dimension and the bipartition of $\omega$ relevant for the definition of $\S^k$ is $AA'|B'B$. 
\end{sdp}
\noindent Note that the first condition in (\ref{newHier}) implies $\tr[\omega]=D$.

\begin{prop}
	$\lim_{k \rightarrow \infty} \S_D^k = \S_D$. More precisely, for any $\sigma_{AB}\in \S_D^k$, there exists $\tilde{\sigma}_{AB}\in \S_D$ such that ${\|\sigma_{AB}-\tilde{\sigma}_{AB}\|_1} \leq O\left(\frac{d^2D^4}{k^2}\right)$, where $d=\min\left(d_A,d_B\right)$.
\end{prop}

\begin{proof}
	Let us first show that $\S_D \subseteq \S_D^k$ for all $k$. Consider $\sigma_{AB}\in \S_D$. Note that any normalized pure state $\ket{\psi}=\sum_{i=1}^Dc_i\ket{\phi_i}\ket{\Phi_i}$ with Schmidt rank at most $D$ can be written as $\ket{\psi}=\Pi^\dagger_D \cdot (\ket{\alpha}_{AA'}\ket{\beta}_{B'B})$, with
	\be
	\ket{\alpha}=\sum_{i=1}^D c_i \ket{\phi_i}_A\ket{i}_{A'},\hspace{2.5mm} \ket{\beta}=\sum_{i=1}^D \ket{i}_{B'}\ket{\Phi_i}_B.
	\ee
	The vector $\ket{\alpha}\ket{\beta}$ has norm $\sqrt{D}$ and is separable with respect to the bipartition $AA'|B'B$. By convexity, we thus have that any $\sigma_{AB}\in \S_D$ can be expressed as $\sigma_{AB}=\Pi^\dagger_D\omega \Pi_D$, for some positive semidefinite $\omega$ that is separable with respect to the bipartition $AA'|B'B$, denoted $\omega \in \S$, and with $\tr[\omega]=D$. Since $\omega /\tr[\omega] \in \S$ implies $\omega /\tr[\omega]\in \S^k$ according to the DPS hierarchy, this implies that $\sigma_{AB} \in \S^k_D$.

	To prove the completeness of the hierarchy (\ref{newHier}), we will invoke known results \cite{NOP, finetti} on the finite approximability of the set of separable states by means of the DPS hierarchy~\cite{DPS1,DPS2,DPS3}. Let $\sigma_{AB} \in \S^k_D$. By \cite{NOP}, we know that there exists a strictly decreasing sequence $(\epsilon_k)_k$, with $\lim_{k\to\infty} \epsilon_k=0$, such that any state $\omega/D \in \S^k$ satisfies
	\be
	\frac{\omega_{sep}}{D}=(1-\epsilon_k)\frac{\omega}{D}+\epsilon_k \frac{\omega'}{D},
	\label{completeness}
	\ee
	\noindent for some state $\omega'/D$ and a separable state $\omega_{sep}/D \in \S$. More specifically, we have $\epsilon_k=O\left(\frac{d^2D^2}{k^2}\right)$ with $d=\min(d_A,d_B)$. For $\sigma_{AB}=\Pi^\dagger_D \omega\Pi_D$ this implies that
	\be
	\Pi^\dagger_D \omega_{sep}\Pi_D=(1-\epsilon_k)\sigma_{AB}+\epsilon_k\Pi^\dagger_D\omega'\Pi_D.
	\label{primer}
	\ee
	Call $\lambda := \tr[\Pi^\dagger_D \omega_{sep}\Pi_D]$ and define the normalized states $\tilde{\sigma} = \Pi^\dagger_D \omega_{sep}\Pi_D/\lambda$, $\sigma' =  \Pi^\dagger_D \omega'\Pi_D /\tr[\Pi^\dagger_D \omega'\Pi_D]$, and the parameter $\epsilon'_k := \epsilon_k\tr[\Pi^\dagger_D \omega'\Pi_D]/\lambda$. Since $\tr[\sigma_{AB}]=1$, we have that
	\be
	\tilde{\sigma}=(1-\epsilon'_k)\sigma_{AB}+\epsilon'_k\sigma',
	\ee
	\noindent where $\tilde{\sigma} \in \mathcal{S}_D$. From the relations $\tr[\omega']=D$ and $\|\Pi_D\Pi^\dagger_D\|_\infty=D$, it follows that $\epsilon'_k\leq \frac{\epsilon_kD^2}{\lambda}$. On the other hand, by \eqref{primer}, we have that $\lambda\geq 1-\epsilon_k$. It follows that $\epsilon'_k\leq \frac{\epsilon_kD^2}{1-\epsilon_k}$. Thus, as $k$ grows, the separation between $\S^k_D$ and $\S_D$ tends to zero as $O\left(\frac{d^2D^4}{k^2}\right)$.
	
\end{proof}

\section{Proof that SDP~2 is an inner approximation to the set of $D$-unfaithful states} \label{app3}

\begin{sdp}
	Let $\rho_{AB}$ be a bipartite state. If there exists $\mu\in [0,1]$ and positive semidefinite operators $M_A$, $M_B$ such that 
	\begin{align} 
		&M_A \otimes \id_B + \id_A \otimes M_B \geq \rho_{AB} \nonumber \\
		&\mu (D-1) = \tr[M_A], \hspace{12mm} \mu\id-M_A\geq 0,\nonumber\\
		&(1-\mu)(D-1)=\tr[M_B], \hspace{2.5mm}(1-\mu)\id-M_B\geq 0,
		\label{criter}
	\end{align}
	\noindent then we say that $\rho_{AB}\in \widetilde{\U}_D$.
\end{sdp}

\begin{prop}
	$\widetilde{\U}_D \subseteq {\U}_D.$
\end{prop}

\begin{proof}
	Let $\rho_{AB} \in \widetilde{\U}_D$. Then for any pure state $\ket{\psi}$ and $\Lambda_{AB}:=\proj{\psi}$,
	
	\be
	\tr[M_A\Lambda_A]+\tr[M_B\Lambda_B]\geq \tr[\rho_{AB}\Lambda_{AB}],
	\label{intermediate}
	\ee
	which is obtained by multiplying $M_A \otimes \id_B + \id_A \otimes M_B \geq \rho_{AB}$ with $\Lambda_{AB}$ from the right and taking the trace.

	Let $\Lambda_A=\sum_{i=1}^{n_A} \lambda_i(\ket{\psi}) \proj{\psi_i}$ denote the spectral decomposition of $\Lambda_A$, where $(\sqrt{\lambda_i(\ket{\psi})})_i$ are the Schmidt coefficients of $\ket{\psi}$ ordered such that $\lambda_1(\ket{\psi}) \geq \lambda_2(\ket{\psi}) \geq \ldots \geq \lambda_{n_A}(\ket{\psi})$ and $n_A$ is the Hilbert space dimension of system $A$. Then, $\tr[M_A\Lambda_A]=\sum_{i=1}^{n_A}\lambda_i(\ket{\psi}) p_i$ for ${p_i:= \bra{\psi_i}M_A\ket{\psi_i}}$. Since $M_A \geq 0$ and $\mu \id \geq M_A$, we have that $0\leq p_i\leq \mu$.  Furthermore, we have 
	\be 
	\mu (D-1) = \tr[M_A] = \sum_{i=1}^{n_A} p_i,
	\label{trA}
	\ee
	since $\{ \ket{\psi_i} \}_i$ is an orthonormal basis. Now, for  $j\geq D$ we have that
	
	\be 
	r_j :=  p_j\frac{\sum_{i=1}^{D-1}(\mu-p_i)\lambda_i(\ket{\psi})}{\sum_{k=1}^{D-1}(\mu-p_k)}-p_j\lambda_j(\ket{\psi})\geq 0,
	\ee
	due to the positivity of the $p_j$ and the ordering of the $\lambda_j(\ket{\psi})$. Hence, it follows that
	\be
	\tr[M_A\Lambda_A]=\sum_{i=1}^{n_A} \lambda_i(\ket{\psi}) p_i \leq \sum_{i=1}^{n_A}\lambda_i(\ket{\psi}) p_i+\sum_{j=D}^{n_A} r_j =\mu\sum_{i=1}^{D-1}\lambda_i(\ket{\psi}),
	\ee
	where the last equality follows since $\sum_{j=D}^{n_A} r_j=\sum_{i=1}^{D-1} \mu \lambda_i(\ket{\psi})-\sum_{i=1}^{n_A}\lambda_i(\ket{\psi}) p_i$, which can be certified with a straightforward calculation using \eqref{trA}.

	With an analogous argument for $\Lambda_B$, we obtain  $\tr[M_B\Lambda_B]\leq (1-\mu)\sum_{i=1}^{D-1}\lambda_i(\ket{\psi})$. Therefore, the left hand side of  \eqref{intermediate} is upper bounded as 
	\be
	\tr[\rho_{AB}\Lambda_{AB}] \leq \sum_{i=1}^{D-1}\lambda_i(\ket{\psi}).
	\label{eqX}
	\ee 
	Recall from the main text, that  $\sum_{i=1}^{D-1}\lambda_i(\ket{\psi})$ is a lower bound on the fidelity $F$ in a dimension-$D$-witness $F \id -\ketbra{\psi}$.
	Since \eqref{eqX} holds for any pure state $\ket{\psi}$, there does not exist a pure-state fidelity, dimension-$D$-witness for $\rho_{AB}$.
	
\end{proof}

\section{Sampling quantum states according to various measures} \label{app4}

For completeness we provide details regarding the algorithms we used for sampling random states for the numerical considerations presented in the main text.

\bigskip

\noindent{\bf Sampling states according to the Hilbert-Schmidt metric}

\smallskip

\noindent To sample a random state $\rho_{AB}$ on an $n$-dimensional Hilbert space according to the Hilbert-Schmidt metric we proceed as follows~\cite{Zyczkowski1, Zyczkowski2}.
\begin{itemize}
	\item Randomly sample two $n \times n$ matrices $M_R$ and $M_I$ by sampling each entry of each matrix from a normal distribution with mean $0$ and variance $1$ and consider $M=M_R+i  M_I$, where $i$ is the imaginary unit.
	\item Take $M'=M M^\dagger$, where $^\dagger$ denotes the conjugate transpose.
	\item  Let $\rho_{AB}=M'/\tr(M')$.	
\end{itemize}

\bigskip

\noindent{\bf Sampling states according to the Bures metric}

\smallskip

\noindent To sample a random state $\rho_{AB}$ on an $n$-dimensional Hilbert space according to the Bures metric we follow the following procedure~\cite{Osipov}.
\begin{itemize}
	\item Randomly sample two $n \times n$ matrices $M_R$ and $M_I$ by sampling each entry of each matrix from a normal distribution with mean $0$ and variance $1$ and consider $M=M_R+i  M_I$, where $i$ is the imaginary unit.
	\item Sample a random $n \times n$ unitary matrix $U$ from the Haar measure.
	\item Take $M'= (\id + U) M M^\dagger (\id + U^\dagger)$, where $^\dagger$ denotes the conjugate transpose.
	\item  Let $\rho_{AB}=M'/\tr(M')$.	
\end{itemize}

\bigskip

\noindent{\bf Sampling real states}

\smallskip

\noindent To sample a random real state $\rho_{AB}$ on an $n$-dimensional Hilbert space we used the following procedure. Note that this is analogous to sampling according to the Hilbert-Schmidt metric but restricted to the real part.
\begin{itemize}
	\item Randomly sample an $n \times n$ matrix $M$ by sampling each entry from a normal distribution with mean $0$ and variance $1$.
	\item Take $M'=M M^T$, where $^T$ denotes the transpose.
	\item  Let $\rho_{AB}=M'/\tr(M')$.	
\end{itemize}

\twocolumngrid

\bibliographystyle{apsrev}

\end{document}